\documentclass[11pt]{article}

\usepackage{amssymb}
\usepackage{amsmath}
\usepackage{amsthm}
\usepackage{textcomp}
\usepackage{array}
\usepackage[latin5]{inputenc}
\usepackage{lineno}
\usepackage{color}
\usepackage[margin=1in]{geometry}
\usepackage{comment}

\usepackage{graphicx}

\newcommand{\ket}[1]{|#1\rangle}
\newcommand{\braket}[2]{\langle #1|#2\rangle}
\newcommand{\cent}[0]{\mbox{\textcent}}
\newcommand{\dollar}[0]{\$}

\newtheorem{theorem}{Theorem}
\newtheorem{corollary}{Corollary}
\newtheorem{lemma}{Lemma}

\theoremstyle{definition}

\newcommand{\setD}{<\mspace{-4mu}>}

\newcommand{\IP}[1]{\mathsf{IP(#1)}}
\newcommand{\IPstar}[1]{\mathsf{IP^{*}(#1)}}

\newcommand{\AM}[1]{\mathsf{AM(#1)}}
\newcommand{\AMstar}[1]{\mathsf{AM^{*}(#1)}}

\newcommand{\dca}{2dca}
\newcommand{\dcca}{1d2ca}
\newcommand{\DCA}{\mathsf{2DCA}}

\newcommand{\nca}{2nca}

\newcommand{\NCA}{\mathsf{2NCA}}
\newcommand{\pca}{2pca}
\newcommand{\onepca}{1pca}
\newcommand{\pfa}{2pfa}
\newcommand{\PCA}{\mathsf{2PCA}}

\newcommand{\qcfa}{2qcfa}
\newcommand{\qcca}{2qcca}
\newcommand{\qca}{2qca}
\newcommand{\rtqca}{rt-qca}

\newcommand{\decidable}{\mathsf{DECIDABLE}}

\newcommand{\siamtwins}{\mathtt{SIAM\mbox{-}TWINS}}
\newcommand{\existentialtwin}{\mathtt{EXIST\mbox{-}TWIN}}
\newcommand{\twin}{\mathtt{TWIN}}

\newcommand{\centerlanguage}{\mathtt{CENTER}}
\newcommand{\greater}{\mathtt{GREATER}}
\newcommand{\greatersquare}{\mathtt{GREATER\mbox{-}SQUARE}}
\newcommand{\usquare}{\mathtt{USQUARE}}
\newcommand{\lapins}{\mathtt{LAPIN\check{S}}}
\newcommand{\squarelanguage}{\mathtt{SQUARE}}
\newcommand{\say}{\mathtt{SAY}}

\title{One-counter verifiers for decidable languages\footnote{This work was partially supported by FP7 FET-Open project QCS.}}
\author{Abuzer Yakary{\i}lmaz \\
\small University of Latvia, Faculty of Computing, Raina bulv. 19, R\={\i}ga, LV-1586, Latvia 
\\
\small \texttt{abuzer@lu.lv}
\\ \\
\today
}

\date{\small Keywords: 
interactive proof systems, 
Arthur-Merlin games, 
decidable languages,
counter automata,
probabilistic and quantum computation}

\begin{document}

\maketitle

\thispagestyle{empty}

\begin{abstract}
Condon and Lipton (FOCS 1989) showed that the class of languages having a space-bounded interactive proof system (IPS) is a proper subset of decidable languages, where the verifier is a probabilistic Turing machine. In this paper, we show that if we use architecturally restricted verifiers instead of restricting the working memory, i.e. replacing the working tape(s) with a single counter, we can define some IPS's for each decidable language. Such verifiers are called two-way probabilistic one-counter automata (2pca's). Then, we show that by adding a fixed-size quantum memory to a 2pca, called a two-way one-counter automaton with quantum and classical states (2qcca), the protocol can be space efficient. As a further result, if the 2qcca can use a quantum counter instead of a classical one, then the protocol can even be public, also known as Arthur-Merlin games. 

We also investigate the computational power of 2pca's and 2qcca's as language recognizers. We show that bounded-error 2pca's can be more powerful than their deterministic counterparts by giving a bounded-error simulation of their nondeterministic counterparts. Then, we present a new programming technique for bounded-error 2qcca's and show that they can recognize a language which seems not to be recognized by any bounded-error 2pca. We also obtain some interesting results for bounded-error 1-pebble quantum finite automata based on this new technique. Lastly, we prove a conjecture posed by Ravikumar (FSTTCS 1992) regarding 1-pebble probabilistic finite automata, i.e. they can recognize some nonstochastic languages with bounded error.
\end{abstract}

%%%%%%%%%%%%%%%%%%%%%%%%%%%%%%%%%%%%%%%%%%%%%%%%%%%%%%%%%%%%%%%%%%%%%%%%%%%%%%%%%%%%%%%%%%%
%%%%%%%%%%%%%%%%%%%%%%%%%%%%%%%%%%%%%%%%%%%%%%%%%%%%%%%%%%%%%%%%%%%%%%%%%%%%%%%%%%%%%%%%%%%
%%%%%%%%%%%%%%%%%%%%%%%%%%%%%%%%%%%%%%%%%%%%%%%%%%%%%%%%%%%%%%%%%%%%%%%%%%%%%%%%%%%%%%%%%%%
\section{Introduction}
\label{sec:introduction}
%%%%%%%%%%%%%%%%%%%%%%%%%%%%%%%%%%%%%%%%%%%%%%%%%%%%%%%%%%%%%%%%%%%%%%%%%%%%%%%%%%%%%%%%%%%
%%%%%%%%%%%%%%%%%%%%%%%%%%%%%%%%%%%%%%%%%%%%%%%%%%%%%%%%%%%%%%%%%%%%%%%%%%%%%%%%%%%%%%%%%%%
%%%%%%%%%%%%%%%%%%%%%%%%%%%%%%%%%%%%%%%%%%%%%%%%%%%%%%%%%%%%%%%%%%%%%%%%%%%%%%%%%%%%%%%%%%%

The only known interactive proof systems (IPS) having a restricted verifier for decidable languages were given by Feige and Shamir \cite{FS89} (and independently by Condon and Lipton \cite{CL89}). Although the verifier is a one-way probabilistic finite automaton, the protocols require to communicate with two provers. Therefore, the question remains open for IPS with one prover. In fact, Condon and Lipton \cite{CL89} showed that the class of languages having a space-bounded interactive proof system (IPS) is a proper subset of decidable languages, where the verifier is a probabilistic Turing machine (PTM). In this paper, we show that if we use architecturally restricted verifiers instead of restricting the working memory, i.e. replacing the working tape(s) with a single counter, we can define some IPS's for each decidable language. 

We present four new protocols for decidable languages. In the first protocol, the verifier is a two-way probabilistic one-counter automaton (\pca). By relaxing the requirement of having arbitrary small error bound, we obtain another protocol, in which the verifier does not need to move its input head to the left. In the third protocol, we show that if the verifier uses a fixed-size quantum register, called a two-way one-counter automaton with quantum and classical states (\qcca), then the protocol can also be space efficient. In all of these protocols, the verifiers hide some information from the prover. In the fourth protocol, we show that if we replace the classical counter of the verifier with a quantum counter, called a two-way quantum one-counter automaton (\qca), then the third protocol turns out to be public (Arthur-Merlin games), i.e. the prover always has a complete information about the verifier. The techniques behind these protocols are inspired from the previous weak-protocols, i.e. the nonmembers does not need to be rejected with high probability by the verifier, given for Turing recognizable languages by Condon and Lipton \cite{CL89} and Yakary{\i}lmaz \cite{Yak12A}.

%If we use a one-way one-way probabilistic one-counter automata (\onepca's) as a verifier, then the error our proto Such verifiers are known as two-way probabilistic one-counter automata (2pca's). If the verifier is forbidden to move its input head to the left (one-way one-way probabilistic one-counter automata (\onepca's)), then the protocol  If we restrict ourselves to one-way probabilistic one-counter automata (\onepca's), we obtain the same result again, but, the error bound cannot be  Moreover, we define two new quantum verifiers: A 2pca augmented with a fixed-size quantum memory, namely two-way quantum automaton with classical counter (2qcca), and a 2qcca having a quantum counter, namely two-way quantum one-counter automaton (2qca). Then, we show that 2qcca's can be more space efficient than 2pca, and the protocol turns out to be public (Arthur-Merlin games) if we use a 2qca as a verifier. 

We also examine \pca's and \qcca's as recognizers. 
%Firstly, we also show that 2pca's can be simulated by linear-space ptm's (and so 2pca's cannot be simulated by any space-bounded ptm's if there is an interaction). Secondly, 
We show that \pca's form a bigger class than two-way deterministic one-counter automata (\dca's). We obtain this result by giving a bounded-error simulation of two-way nondeterministic one-counter automata. Then, we present a new programming technique for bounded-error \qcca's and show that they can recognize a language which seems not to be recognized by any bounded-error \pca. We also obtain some interesting results for bounded-error 1-pebble quantum finite automata based on this new technique. Moreover, we prove a conjecture posed by Ravikumar \cite{Rav92} regarding 1-pebble probabilistic finite automata, i.e. they can recognize some nonstochastic languages with bounded error.

To our knowledge, \pca's have never been investigated before. The only related work that we know is \cite{HS05}, in which  Hromkovic and Schnitger examined two-way probabilistic multi-counter machines that are restricted to polynomial time, but, they did not present any result related to \pca's. \qca's, on the other hand, are examined only by Yamasaki at. al. \cite{YKI05} as language recognizers, in which the authors presented some bounded-error \mbox{\qca} algorithms for some languages. However, these languages are known to be recognized by \dca's and bounded-error \qcca's without a counter. Therefore, our results seem the first interesting results on \qca's. 

We provide the necessary background in Section \ref{sec:background}. The four protocols for decidable languages are given in Section \ref{sec:verifier}. The language recognition powers of probabilistic and quantum counter automata are investigated in Section \ref{sec:recognizer}. Lastly, the results on probabilistic and quantum finite automata with 1-pebble are given in Section \ref{sec:pebble-automata}.

%%%%%%%%%%%%%%%%%%%%%%%%%%%%%%%%%%%%%%%%%%%%%%%%%%%%%%%%%%%%%%%%%%%%%%%%%%%%%%%%%%%%%%%%%%%
%%%%%%%%%%%%%%%%%%%%%%%%%%%%%%%%%%%%%%%%%%%%%%%%%%%%%%%%%%%%%%%%%%%%%%%%%%%%%%%%%%%%%%%%%%%
%%%%%%%%%%%%%%%%%%%%%%%%%%%%%%%%%%%%%%%%%%%%%%%%%%%%%%%%%%%%%%%%%%%%%%%%%%%%%%%%%%%%%%%%%%%
\section{Background}
\label{sec:background}
%%%%%%%%%%%%%%%%%%%%%%%%%%%%%%%%%%%%%%%%%%%%%%%%%%%%%%%%%%%%%%%%%%%%%%%%%%%%%%%%%%%%%%%%%%%
%%%%%%%%%%%%%%%%%%%%%%%%%%%%%%%%%%%%%%%%%%%%%%%%%%%%%%%%%%%%%%%%%%%%%%%%%%%%%%%%%%%%%%%%%%%
%%%%%%%%%%%%%%%%%%%%%%%%%%%%%%%%%%%%%%%%%%%%%%%%%%%%%%%%%%%%%%%%%%%%%%%%%%%%%%%%%%%%%%%%%%%

Throughout the paper, $ \Sigma $ not containing $ \cent $ and $ \dollar $ denotes the input alphabet and $ \tilde{\Sigma} = \Sigma \cup \{ \cent,\dollar \} $. For a given string $ x $, $ |x| $ is the length of $ x $ and $ x_{i} $ is the $ i^{th} $ symbol of $ x $, where $ 1 \leq i \leq |x| $. The string $ \cent x \dollar $ is represented by $ \tilde{x} $. Moreover, $ \setD $ is the set of $ \{ \leftarrow,\downarrow,\rightarrow \} $, $ \Theta $ is the set of $ \{0,\pm\} $, and $ \diamondsuit $ is the set of $ \{-1,0,1\} $. $ \mathbf{P}(\cdot) $ denotes all subsets of a given set.

Each model defined in the paper has a two-way infinite read-only input tape whose squares are indexed by integers. Any given input string, say $ x \in \Sigma^{*} $, is placed on the tape as $ \tilde{x} $ between the squares indexed by 1 and $ | \tilde{x} | $. The tape has a single head, and it can stay in the same position ($ \downarrow $) or move to one square to the left ($ \leftarrow $) or to the right ($ \rightarrow $) in one step. It must always be guaranteed that the input head never leaves $ \tilde{w} $. Some models in the paper have also a counter, an infinite storage having two status, i.e. \textit{zero} ($ 0 $) or \textit{nonzero} ($ \pm $), and being updated by a value from $ \diamondsuit $ in one step. We assume that the reader familar with the modes of language recognition with errors (see also Appendix \ref{app:language-recognition-mode}).

%%%%%%%%%%%%%%%%%%%%%%%%%%%%%%%%%%%%%%%%%%%%%%%%%%%%%%%%%%%%%%%%%%%%%%%%%%%%%%%%%%%%%%%%%%%
%%%%%%%%%%%%%%%%%%%%%%%%%%%%%%%%%%%%%%%%%%%%%%%%%%%%%%%%%%%%%%%%%%%%%%%%%%%%%%%%%%%%%%%%%%%
\subsection{Classical models}
\label{sec:classical-models}
%%%%%%%%%%%%%%%%%%%%%%%%%%%%%%%%%%%%%%%%%%%%%%%%%%%%%%%%%%%%%%%%%%%%%%%%%%%%%%%%%%%%%%%%%%%
%%%%%%%%%%%%%%%%%%%%%%%%%%%%%%%%%%%%%%%%%%%%%%%%%%%%%%%%%%%%%%%%%%%%%%%%%%%%%%%%%%%%%%%%%%%

A two-way deterministic one-counter automaton (\dca) is a two-way deterministic finite automaton with a counter. Formally, a \mbox{\dca} $ \mathcal{D} $ is a 6-tuple $ \mathcal{D} = (S,\Sigma,\delta,s_{1},s_{a},s_{r}) $, where $ S $ is the set of states, $ s_1 \in Q $ is the initial state, $ s_a \in S $ and $ s_r \in S $ ($ s_a \neq s_r $) are the accepting and rejecting states, respectively, and $ \delta $  is the transition function governing the behaviour of $ \mathcal{D} $ in each step, i.e. $ \delta: S \times \tilde{\Sigma} \times \Theta \rightarrow S \times \mspace{-5mu} \setD \mspace{-5mu} \times \diamondsuit. $ Specifically, $ \delta(s,\sigma,\theta) \rightarrow (s',d_i,c) $ means that when $ \mathcal{D} $ is in state $ s \in S $, reads symbol $ \sigma \in \tilde{\Sigma} $, and the status of its counter is $ \theta \in \Theta $, then it updates its state to $ s' \in S $, the position of the input head with respect to $ d_i \in \setD $, and the value of the counter by $ c \in \diamondsuit $.

At the beginning of the computation, $ \mathcal{D} $ is in state $ s_1 $, the input head is placed on symbol $ \cent $, and the value of the counter is set to zero. A configuration of $ \mathcal{D} $ on a given input string is represented by a triple $ (s,i,v) $, where $ s $ is the state, $ i $ is the position of the input head, and $ v $ is the value of the counter. The computation is terminated and the input is accepted (rejected) by $ \mathcal{D} $ when it enters to $ s_a $ ($ s_r $). The class of languages recognized by \dca's is denoted $ \DCA $.

We will use \textit{the same terminology} also for the other models unless otherwise is specified. A two-way nondeterministic one-counter automaton (\nca), say $ \mathcal{N} $, is a \mbox{\dca} having capability of making nondeterministic choices in each step. The transition function of $ \mathcal{N} $ is extended as follows: $ \delta: S \times \tilde{\Sigma} \times \Theta \rightarrow \mathbf{P} ( S \times \mspace{-5mu} \setD \mspace{-5mu} \times \diamondsuit) $. In other words, for each triple $ (s,\sigma,\theta) $ (see above), there may be more than one transition. Thus, $ \mathcal{N} $ can follow more than one path during the computation, and if any path ends with a decision of acceptance, then the input is accepted. The class of languages recognized by \nca's is denoted $ \NCA $.  % If we remove the counter, then we obtain a \nfa (two-way nondeterministic finite automaton).

A two-way probabilistic one-counter automaton (\pca), say $ \mathcal{P} $, is a \mbox{\dca} having capability of making probabilistic choices in each step. In order to explicitly represent the probabilistic part the machine, each step is divided into two transitions. Formally, $ \delta = (\delta_{p},\delta_d) $. For each triple $ (s,\sigma,\theta) $ (see above), there are some predefined outcomes, i.e. $ \Delta_{(s,\sigma,\theta)} = \{1,\ldots,k_{(s,\sigma,\theta)}\} $. Each outcome is selected with some (rational) probability: $ \delta_p(s,\sigma,\theta,\tau) = p_{\tau} \in \mathbb{Q} $, where $ \tau \in \Delta $ and $ \sum_{\tau \in \Delta} p_{\tau} = 1 $. After observing the outcome ($ \tau $), the deterministic transition is implemented as follows:  $ \delta_d(s,\sigma,\theta) \overset{\tau} \rightarrow (s',d_i,c) $ (see above). Note that $ \delta_d $ must be defined for each possible $ \tau $. 
%Each path of $ \mathcal{P} $ can be followed with some probabilities. A language $ \mathcal{L} $ is recognized by $ \mathcal{P} $ with bounded error if each $ x \in \mathtt{L} $ is accepted with a probability at least $ 1-\epsilon $ and each $ x \notin \mathtt{L} $ is rejected with a probability at least $ 1-\epsilon $, where $ \epsilon < \frac{1}{2} $ is the error bound. 
The class of languages recognized by \pca's with bounded error is denoted $ \PCA $. If we remove the counter, then we obtain a \mbox{\pfa} (two-way probabilistic finite automaton). If the input head of a \mbox{\pca} is not allowed to move to left, then we obtain a one-way probabilistic one-counter automaton (\onepca). 

A one-way deterministic two-counter automaton (\dcca) is a one-way deterministic finite automaton with two counters. It was shown that any deterministic Turing machine (DTM) can be simulated by a \dcca \cite{Mi61}. We denote the class of decidable languages $ \decidable $.

%%%%%%%%%%%%%%%%%%%%%%%%%%%%%%%%%%%%%%%%%%%%%%%%%%%%%%%%%%%%%%%%%%%%%%%%%%%%%%%%%%%%%%%%%%%
%%%%%%%%%%%%%%%%%%%%%%%%%%%%%%%%%%%%%%%%%%%%%%%%%%%%%%%%%%%%%%%%%%%%%%%%%%%%%%%%%%%%%%%%%%%
\subsection{Quantum models}
\label{sec:quantum-models}
%%%%%%%%%%%%%%%%%%%%%%%%%%%%%%%%%%%%%%%%%%%%%%%%%%%%%%%%%%%%%%%%%%%%%%%%%%%%%%%%%%%%%%%%%%%
%%%%%%%%%%%%%%%%%%%%%%%%%%%%%%%%%%%%%%%%%%%%%%%%%%%%%%%%%%%%%%%%%%%%%%%%%%%%%%%%%%%%%%%%%%%

A two-way finite state automaton with quantum and classical states \cite{AW02,YS11A} (\qcfa) is a \mbox{\pfa} using a finite quantum register instead of a classical random generator. Note that the quantum register can keep some information by its (pure) quantum state as well as making probabilistic choices.\footnote{ It was shown that \qcfa's are more powerful than \pfa's by Ambainis and Watrous \cite{AW02}. Moreover, Yakary{\i}lmaz and Say \cite{YS10A,YS10B} showed that they can also recognize many interesting languages.}

Formally, a \mbox{\qcfa}\footnote{Although the formal definition of \mbox{\qcfa} given here is a bit different than the ones given in \cite{AW02,YS11A}, all the models are equivalent.} $ \mathcal{Q} $ is a 8 tuple $ (S,Q,\Sigma,\delta,s_1,s_a,s_r,q_1) $, where, apart from a classical model, there are two different components: $ Q $ is the state set of quantum register and $ q_1 $ is its initial state. Similar to probabilistic models, $ \delta = (\delta_q,\delta_d) $, where $ \delta_q $ governing the quantum part. In each step, firstly, $ \delta_q $ determines a superoperator (see Figure \ref{fig:superoperators} for the details) depending on the current classical state ($ s \in S $) and scanning symbol ($ \sigma \in \tilde{\Sigma} $), i.e. $ \mathcal{E}_{s,\sigma} $, and then it is applied to the quantum register and one outcome, say $ \tau $, is observed. Secondly, the classical part of $ \mathcal{Q} $ is updated depending on $ s $, $ \sigma $, and $ \tau $, which is formally represented as $ \delta_c(s,\sigma) \overset{\tau} \rightarrow (s',d_i) $, where $ s' \in S $ is the new classical state and $ d_i \in \setD $ is the update of the position of input tape. Note that $ \delta_d $ must be defined for each possible $ \tau $. 
% The class of languages recognized by \qcfa's with bounded error is denoted $ \QCFA $.

\begin{figure}[!ht]
	\centering
	\footnotesize
	\fbox{
	\begin{minipage}{0.95\textwidth}
		The most general quantum operator is a superoperator,
		which generalizes stochastic and unitary operators and also includes measurement.
		Formally, a superoperator $ \mathcal{E} $ is composed by a finite number of operation elements,
		$ \mathcal{E} = \{ E_{1}, \ldots, E_{k} \} $, satisfying that
		\begin{equation}
		\label{eq:completeness}
			\sum_{i=1}^{k} E_{i}^{\dagger} E_{i} = I,
		\end{equation}
		where $ k \in \mathbb{Z}^{+} $ and the indices are the measurement outcomes.
		When a superoperator, say $\mathcal{E}$, is applied to 
		the quantum register in state $\ket{\psi} $, i.e. $ \mathcal{E}(\ket{\psi}) $,
		we obtain the measurement outcome $i$ with probability 
		$ p_{i} = \braket{\widetilde{\psi_{i}}}{\widetilde{\psi_{i}}} $,
		where $\ket{\widetilde{\psi_{i}}}$, \textit{the unconditional state vector}, 
		is calculated as $ \ket{\widetilde{\psi}_{i}} = E_{i} \ket{\psi} $ and $1 \leq i \leq k$.
		(Note that using unconditional state vector simplifies calculations in many cases.)
		If the outcome $i$ is observed ($p_{i} > 0 $), the new state of the system 
		is obtained by normalizing $ \ket{\widetilde{\psi}_{i}} $, 
		which is $ \ket{\psi_{i}} = \frac{\ket{\widetilde{\psi_{i}}}}{\sqrt{p_{i}}} $.
		Moreover, as a special operator, the quantum register can be initialized to a predefined quantum state.
		This initialize operator, which has only one outcome, is denoted $ \acute{\mathcal{E}} $.
		In this paper, the entries of quantum operators are defined by rational numbers.
		Thus the probabilities of the outcomes are always rational numbers.
	\end{minipage}
	}
	\caption{The details of superoperators \cite{Yak12A}}
	\label{fig:superoperators}
\end{figure}

A two-way one-counter automaton with quantum and classical states (\qcca) is a \mbox{\pca} using a finite quantum register instead of a classical random generator. The formal definition of a \mbox{\qcca} is exactly the same as a \qcfa. In fact, a \mbox{\qcca} is a \mbox{\qcfa} with a classical counter. So, the transition functions of a \mbox{\qcfa} ($ \delta_q $ and $ \delta_c $) can be extended for a \mbox{\qcca} with the following modifications:
\begin{itemize}
	\item The superoperator is determined by also the status of the counter ($ \theta \in \Theta $), i.e. $ \mathcal{E}_{s,\sigma,\theta} $.
	\item The classical part of $ \mathcal{Q} $ is updated depending on $ s $, $ \sigma $, $ \theta $, and $ \tau $, which is formally represented as $ \delta_c(s,\sigma,\theta) \overset{\tau} \rightarrow (s',d_i,c) $, where $ s' \in S $ is the new classical state, $ d_i $ is the update of the position of input tape, and $ c \in \diamondsuit  $ is update on the counter.
\end{itemize}  
%The class of languages recognized by \qcca's with bounded error is denoted $ \QCCA $.

A generalization of \mbox{\qcca} is a two-way quantum counter automaton (\qca) that uses a quantum counter instead of a classical one. (Note that this model is still not the most general one, but it is sufficiently general for our purpose.) We can see \mbox{\qca} as the combination of a \mbox{\qcfa} and a realtime quantum one-counter automaton (\rtqca) \cite{SY12A}: The \mbox{\qcfa} part governs the computation, and access the counter through the \mbox{\rtqca} by feeding some input to it and also observing the outcomes. We will use this model in one of our results (Theorem \ref{thm:am-2qca}), and our simple definition will also simplify the proof.
%(We refer the reader to Appendix XX for the definition of \mbox{\rtqca}.)

%%%%%%%%%%%%%%%%%%%%%%%%%%%%%%%%%%%%%%%%%%%%%%%%%%%%%%%%%%%%%%%%%%%%%%%%%%%%%%%%%%%%%%%%%%%
%%%%%%%%%%%%%%%%%%%%%%%%%%%%%%%%%%%%%%%%%%%%%%%%%%%%%%%%%%%%%%%%%%%%%%%%%%%%%%%%%%%%%%%%%%%
\subsection{Interactive proof systems}
\label{sec:ips}
%%%%%%%%%%%%%%%%%%%%%%%%%%%%%%%%%%%%%%%%%%%%%%%%%%%%%%%%%%%%%%%%%%%%%%%%%%%%%%%%%%%%%%%%%%%
%%%%%%%%%%%%%%%%%%%%%%%%%%%%%%%%%%%%%%%%%%%%%%%%%%%%%%%%%%%%%%%%%%%%%%%%%%%%%%%%%%%%%%%%%%%

In this part, we provide the necessary background, based on \cite{DS92,Co93A}, for the proof systems.
%For a detailed survey on space-bounded interactive proof systems, we refer the reader to \cite{Co93A}.
An interactive proof system (IPS)  consists of a prover ($ P $) and a verifier ($ V $). The verifier is a restricted/resource-bounded machine. The classical states of the verifier are partitioned into reading, communication, and halting (accepting or rejecting) states, and it has a special communication cell for communicating with the prover, where the capacity of the cell is finite.

The one-step transitions of the verifier can be described as follows. When in a reading state, the verifier implements its standard transition. When in a communication symbol, the verifiers writes a symbol on the communication cell with respect to the current state. Then, in response, the prover writes a symbol in the cell. Based on the state and the symbol written by prover, the verifier defines the next state of the verifier. Note that the communication is always classical even though the verifier can use some quantum memory.

The prover $ P $ is specified by a prover transition function, which determines the response of the prover to the verifier based on the input and the verifier's communication history until then. Note that this function does not need to be \textit{computable}.

%For a given input $ x $, the probability that $ (P,V) $ accepts (rejects) $ x $ is the cumulative accepting (rejecting) probabilities taken over all branches of the verifier. 
The prover-verifier pair $ (P,V) $ is an IPS for language $ \mathtt{L} $ with error probability $ \epsilon < \frac{1}{2} $ if
(i) for all $ x \in \mathtt{L} $, the probability that $ (P,V) $ accepts $ x $ is greater than $ 1-\epsilon $,
(ii) for all $ x \notin \mathtt{L} $, and all provers $ P^* $, the probability that $ (P^*,V) $ rejects $ x $ is greater than $ 1-\epsilon $.
These conditions are known as completeness and soundness, respectively. 

An Arthur-Merlin (AM) proof system is a special case of IPS such that after each probabilistic or quantum operation, the outcome is automatically written on the communication cell, and so the prover can have
complete information about the computation of the verifier.\footnote{Note that all the verifiers defined in the paper are  allowed to use only rational number transitions.} We also refer them as public proof systems. 

$ \IP{v} $ represents the class of languages having an IPS with some error probability $ \epsilon < \frac{1}{2} $, where the verifier is $ v $-type. Moreover, $ \IPstar{v} $ is a subset of $ \IP{v} $ providing that each language in $ \IPstar{v} $ has an IPS for any error bound. $ \AM{v} $ and  $ \AMstar{v} $ are defined similarly.
%We will use $ \IP{v} $ and $ \AM{v} $ to represent complexity classes for IP and AM systems having $ v $-type verifier, respectively. Moreover, $ \IPstar{v} $ is a subset of $ \IP{v} $ such that each language in $ \IPstar{v} $ has an IPS for any error bound. $ \AMstar{v} $ can be defined similarly.

%%%%%%%%%%%%%%%%%%%%%%%%%%%%%%%%%%%%%%%%%%%%%%%%%%%%%%%%%%%%%%%%%%%%%%%%%%%%%%%%%%%%%%%%%%%
%%%%%%%%%%%%%%%%%%%%%%%%%%%%%%%%%%%%%%%%%%%%%%%%%%%%%%%%%%%%%%%%%%%%%%%%%%%%%%%%%%%%%%%%%%%
%%%%%%%%%%%%%%%%%%%%%%%%%%%%%%%%%%%%%%%%%%%%%%%%%%%%%%%%%%%%%%%%%%%%%%%%%%%%%%%%%%%%%%%%%%%
\section{Counter automata verifiers for decidable languages}
\label{sec:verifier}
%%%%%%%%%%%%%%%%%%%%%%%%%%%%%%%%%%%%%%%%%%%%%%%%%%%%%%%%%%%%%%%%%%%%%%%%%%%%%%%%%%%%%%%%%%%
%%%%%%%%%%%%%%%%%%%%%%%%%%%%%%%%%%%%%%%%%%%%%%%%%%%%%%%%%%%%%%%%%%%%%%%%%%%%%%%%%%%%%%%%%%%
%%%%%%%%%%%%%%%%%%%%%%%%%%%%%%%%%%%%%%%%%%%%%%%%%%%%%%%%%%%%%%%%%%%%%%%%%%%%%%%%%%%%%%%%%%%

In this section, we will present four different protocols for decidable languages. We begin with the classical verifiers.

\begin{theorem}
	\label{thm:ip-star-2pca}
	$ \mathsf{IP^{*}(\pca)} = \decidable $.
\end{theorem}
\begin{proof}
	The relation $ \mathsf{IP(2pca)} \subseteq \mathsf{DECIDABLE} $ is trivial. We will give the proof for the other direction. The proof idea is inspired from the protocol given by Condon and Lipton \cite{CL89}.
	
	Let $ \mathtt{L} $ be a decidable language. Then there exists a \mbox{\dcca} $ \mathcal{D} $, which halts on every input, recognizing $ \mathtt{L} $ \cite{Mi67}. Any configuration of $ \mathcal{D} $ on an input, say $ x \in \Sigma^{*} $, can be represented by $ (s,i,u,v) $, where $ s $ is the state, $ i $ is the head position, and $ u $ and $ v $ are contents of the counters.
	
	We will describe an IPS $ (P,V) $ for $ \mathtt{L} $ by giving a simulation of $ \mathcal{D} $ on the given input, say $ x $, where $ V $ is a $ \pca $. If $ V $ accesses the status of both counters in each step, then it can easily simulate $ \mathcal{D} $ on $ x $ by tracing the state and the head position updates of $ \mathcal{D} $. The prover can provide the contents of the counters for each step. But, the verifier should be careful about the cheating provers. For this purpose, $ V $ can use its counter. That is, in each step, the verifier can determine the changes on the counters, and so can compare the current value and the next value of a counter. Therefore, before starting the simulation, $ V $ equiprobably selects a counter of $ \mathcal{D} $ to test the changes on it. Moreover, $ V $ should also compare the contents of the selected counter not only for $ (2i-1)^{th} $ and $ (2i)^{th} $ steps but also for $ (2i)^{th} $ and $ (2i+1)^{th} $ steps, where $ i \geq 1 $. Thus, $ V $ can start the comparisons from either the first step or the second step, which can also be decided equiprobably at the beginning of the simulation. Therefore, we can identify four comparison strategies, i.e. $ C_{i}^{j} $ ($ V $ selects the $ i^{th} $ counter of $ \mathcal{D} $ and starts to compare from the step-$ j $), where $ 1 \leq i,j \leq 2 $. This is the base strategy of $ V $. However, as described below, it is not sufficient to define a protocol for any error bound.
	
	The simulation of $ \mathcal{D} $ on $ x $ by $ (P,V) $ is executed in an infinite loop. In each round, a new simulation is started. $ V $ requests the contents of the counters for each step from the prover. Let $ w $ be the string obtained from the prover in a single round. The verifier expects $ w $ as $ a^{u_1}b^{v_1}\#a^{u_2}b^{v_2} \# \cdots \# a^{u_t}b^{v_t}\# $ such that $ u_j $ ($v_j  $) is the content of the first (the second) counter after $ j^{th} $ step, where $ 1 \leq j \leq t $ and $ t \geq 1 $. On the other hand, there are four disjoint cases for $ w $ as listed below:
	\begin{itemize}
		\item (C1) $ w $ is of the form \fbox{$ (a^{*}b^{*}\#)^{+} $},
		\item (C2) there is an $ a $ after $ b $ in $ w $,
		\item (C3) $ w $ is infinite and of the form \fbox{$ (a^{*}b^{*}\#)^{+} aaa \cdots $} or \fbox{$ (a^{*}b^{*}\#)^{+} a^{*} bbb \cdots $}, or
		\item (C4) $ w $ is infinite and of the form \fbox{$ (a^{*}b^{*}\#)(a^{*}b^{*}\#)(a^{*}b^{*}\#)\cdots $}.
	\end{itemize}
	It is obvious that $ V $ can check C2 deterministically, and reject the input if there exists an $ a $ after $ b $ in $ w $. In other words, such a round is certainly terminated with the decision of rejection. Therefore, in the remaining part, we assume that $ w $ satisfies one of the other cases. 
	
	If $ w $ is valid (correct), then $ V $ can exactly simulate $ \mathcal{D} $ on $ x $. Otherwise, the simulation may contain some defects, and so $ V $ may give a wrong decision. Moreover, $ V $ may also enter an infinite loop. Note that since $ P $ is honest and provides the valid $ w $, we specifically focus on the strategies of cheating provers on the nonmembers: In each round, $ V $ should deal with infinite loops and should also guarantee that, for the nonmembers, the probability of accepting the input, which can only be given based on the simulation, is sufficiently smaller than the probability of rejecting the input due to detecting the defects on $ w $. The aforementioned (base) strategy of $ V $ is quite strong, and so any invalid $ w $ is detected by at least one of $ C_{i}^{j} $. But the prover can still mislead the verifier in the other choices. Thus, the defect can be detected with a probability at least $ \frac{1}{4} $, and the verifier can follow an invalid $ w $ with a probability at most $ \frac{3}{4} $. Therefore, when $ V $ is convinced to accept the input, it gives the decision of acceptance with probability $ \frac{1}{k} $, and terminates the current round with the remaining probability $ 1-\frac{1}{k} $. So, the total accepting probability of an invalid computation ($ \frac{3}{4k} $) can be sufficiently small compared to the rejecting probability due to the defect ($ \frac{1}{4} $) by setting $ k $ to an appropriate value. However, there is still the problem of infinite loop. We can solve this problem by terminating the round with probability $ \frac{1}{2} $ after obtaining a symbol $ w $ from the prover. Thus, any infinite loop can be terminated with probability 1. Although the probability of making decisions is dramatically decreased due to this new restart strategy, the ratio of accepting and rejecting probabilities for the nonmembers can still be preserved since any decision of acceptance can only be given after a defect. 
	
	Now, we can analyse the overall protocols. Let $ l $ be the length of the valid $ w $. If $ x \in \mathtt{L} $, then $ V $ accepts $ x $ with probability $ \frac{1}{k2^{l}} $ in each round, and so it is accepted exactly. If $ x \notin \mathtt{L} $, if there is no defect, then it is rejected with probability $ \frac{1}{2^{l}} $ in a single round. If there is a defect, than it is detected by $ V $ after obtaining $ (l_1)^{th} $ symbol of $ w $, where $ l_1 \leq l $. Moreover, the input can be accepted by $ V $ after obtaining $ l_2 \geq l_1 $ symbol of $ w $. Then, the input is rejected with a probability at least $ \frac{1}{42^{l_1}} $, and it is accepted with a probability at most $ \frac{3}{4k2^{l_2}} $. Thus, the input is rejected with high probability depending on the value of $ k $. Moreover, the protocol is always terminated with probability 1.
\end{proof}

In the protocol above, if we allow the infinite loops, we can still obtain an IPS for any decidable language by using the base strategy. Besides, it is sufficient to simulate $ \mathcal{D} $ once. Thus, the verifier does not need to move its input head to the left.

\begin{corollary}
	\label{cor:ip-2pca}
	$ \mathsf{IP(\onepca)} = \decidable $.
\end{corollary}
\begin{proof}
	The input is rejected with probability $ \frac{3}{7} $ by the verifier at the beginning of the computation. Then the verifier follows its base strategy once. Therefore, the members are accepted with probability $ \frac{4}{7} $ by the help of a honest prover, and the non-members are rejected with a probability at least $ \frac{3}{7} + \frac{4}{7} \left( \frac{1}{4} \right) = \frac{4}{7} $. The error bound is $ \frac{3}{7} < \frac{1}{2} $.
\end{proof}

%\begin{openproblem}
%	What is the computational power of $ \mathsf{AM(\pca)} $?
%\end{openproblem}

We continue with the quantum verifiers. Recently, Yakary{\i}lmaz \cite{Yak12A} showed that for each Turing-recognizable language, say $ \mathtt{L} $, there exists an AM proof systems with a \mbox{\qcfa} verifier, say $ (P,V) $, such that each $ x \in \mathtt{L} $ is accepted by $ V $ exactly and each $ x \notin \mathtt{L} $ is accepted with a small probability. Such proof systems are also known as weak-IPS \cite{DS92,Co93A}.

By combining the protocol given in \cite{Yak12A} with the first protocol given above (given in the proof of Theorem \ref{thm:ip-star-2pca}), we present two more protocols for decidable languages. A review of the protocol given in \cite{Yak12A} is as follows. Let $ \mathtt{L} $ be a decidable language, and $ \mathcal{D} $ be a DTM, which halts on every input, recognizing $ \mathtt{L} $. In this protocol $ (P,V) $ simulates the computation of $ \mathcal{D} $ on a given input, say $ x $. In an infinite loop, $ V $ requests the computation of $ \mathcal{D} $  on $ x $, as $ w = c_1 \dollar \dollar c_2 \dollar \dollar c_3 \dollar \dollar \cdots $, where $ c_{i > 0} $'s are some configurations of $ \mathcal{D} $ on $ x $ and $ c_1 $ is the initial one. If $ x \in \mathtt{L} $, $ P $ provides the valid $ w $, and $ V $ accepts the input with some probability in each round, then it is accepted exactly. If $ x \notin \mathtt{L} $, then the input is always rejected with a bigger probability than the accepting probability in a single round as long as the prover sends $ \dollar \dollar $ symbols. If the prover does not send  $ \dollar \dollar $ after some point, the round is still terminated with probability 1, but probably with no decision. This is why the system is ``weak''. From a given configuration, the length of the next valid configuration can be easily determined, which can be differ at most one. So, if the verifier uses a classical counter, it can detect when the prover does not send $ \dollar \dollar $ with some probability, i.e. instead of terminating the round with ``no decision'', the round is terminated with some nonzero ``rejecting'' probability. Similar to the first protocol given above, the verifier equiprobably decides to compare the lengths of which configurations, i.e. $ (2i-1)^{th} $ and $ (2i)^{th} $ configurations or $ (2i)^{th} $ and $ (2i+1)^{th} $ configurations, at the beginning of each round, where $ i \geq 1 $. Although the protocol given in \cite{Yak12A} is a public one, the computations on the classical counter must be hidden from the prover in the new protocol. We can formalize this result as follows.

\begin{theorem}
	\label{thm:ip-star-2qcca}
	$ \IPstar{\qcca} = \decidable $.
\end{theorem}

It is a well-known fact that the simulation of a DTM by a \mbox{\dcca} is space (and time) inefficient \cite{vEB90}. Therefore, we can say that for the same language, the latter protocol can be more space efficient than the former protocol for the members of the language since the latter protocol directly simulates a DTM. This can be seen as an advantage of using a few quantum states. 

\begin{corollary}
	\label{cor:ip-2qcca}
	For any language recognized by a $ s(n) $-space DTM, there exists an IPS with a \mbox{\qcca} verifier such that the verifier uses $ s(n) $-space on its counter for the members.
\end{corollary}

Our fourth result is to make the third protocol (given for \qcca) \textit{public}. This can be achieved by replacing the classical counter with a quantum counter. The private part of the the latter protocol is to hide the probabilistic choice at the beginning of each round, based on which the verifier decides the lengths of which configurations will be compared. Since a quantum memory can be in superposition of more than one classical configuration, a \mbox{\qca} can parallelly implement both choices in a public manner.

\begin{theorem}
	\label{thm:am-2qca}
	$ \AMstar{\qca} = \decidable $.
\end{theorem}
\begin{proof}
	The protocol is exactly the same as the third protocol except the counter operations. Therefore, we explain only this part. As mentioned in Section \ref{sec:background}, we can see the verifier as the combination of a \mbox{\qcfa}  and a \rtqca. Remember that the verifier requests $ w = c_1 \dollar \dollar c_2 \dollar \dollar c_3 \dollar \dollar \cdots $ from the prover in each round. So, as long as getting $ w $, the \mbox{\qcfa} part can feed the following sequence $ u = a^{|c_1|}\#i_1\#a^{|c_2|}\#i_2\#a^{|c_3|}\#i_3\#\cdots $ to \mbox{\rtqca} part, where $ i_j $ represents the expected change in the length of $  c_{j+1} $ based on $ c_j $ and $ j>0 $. It is not hard to show that a \mbox{\rtqca} can check the equalities $ |c_{j+1}| = |c_j|+i_j $ for each $ j>0 $, and can also detect the case of $ |c_{j+1}| > |c_j|+i_j $ with some probability.
\end{proof}

%%%%%%%%%%%%%%%%%%%%%%%%%%%%%%%%%%%%%%%%%%%%%%%%%%%%%%%%%%%%%%%%%%%%%%%%%%%%%%%%%%%%%%%%%%%
%%%%%%%%%%%%%%%%%%%%%%%%%%%%%%%%%%%%%%%%%%%%%%%%%%%%%%%%%%%%%%%%%%%%%%%%%%%%%%%%%%%%%%%%%%%
%%%%%%%%%%%%%%%%%%%%%%%%%%%%%%%%%%%%%%%%%%%%%%%%%%%%%%%%%%%%%%%%%%%%%%%%%%%%%%%%%%%%%%%%%%%
\section{Counter machines as recognizer}
\label{sec:recognizer}
%%%%%%%%%%%%%%%%%%%%%%%%%%%%%%%%%%%%%%%%%%%%%%%%%%%%%%%%%%%%%%%%%%%%%%%%%%%%%%%%%%%%%%%%%%%
%%%%%%%%%%%%%%%%%%%%%%%%%%%%%%%%%%%%%%%%%%%%%%%%%%%%%%%%%%%%%%%%%%%%%%%%%%%%%%%%%%%%%%%%%%%
%%%%%%%%%%%%%%%%%%%%%%%%%%%%%%%%%%%%%%%%%%%%%%%%%%%%%%%%%%%%%%%%%%%%%%%%%%%%%%%%%%%%%%%%%%%

In this section, we examine the bounded-error computational powers of \pca's and \qcca's as language recognizers. We begin with  a useful lemma and a lower bound to $ \PCA $.

\begin{lemma}
	\label{lemma:2nca-shortest-path}
	Let $ \mathcal{N} = (S,\Sigma,\delta,s_1,s_a,s_r) $ be a \mbox{\nca}, $ x $ be an input, and $ M = |S||\tilde{x}| $. If  $ s \in S $ is reachable from $ s_1 $ by $ \mathcal{N} $ on $ x $, then there is a path of length no more than $ M^{2} $ from $ (s_1,1,0)  $ to $ (s,i,u) $ for some $ 1 \leq i \leq |\tilde{x}| $ and $ u \leq M $ such that the value of counter never exceeds $ M $. 
\end{lemma}
\begin{proof}
	Let $ path(s_1,s) $ be a path from $ (s_1,1,0)  $ to $ (s,i,u) $ for some $ 1 \leq i \leq |\tilde{x}| $ and $ u \geq 0 $. We can assume that there is no two configurations $ (s',i',u_1>0) $ and $ (s',i',u_2 > u_1) $ in $ path(s_1,s) $ such that the latter one comes after the first one and the counter is never set to zero in between. Let $ T  $ be the set of $ S \times \{1,\ldots,|\tilde{x}|\} $. (Note that $ |T| = M $.) Consider a subpath of $ path(s_1,s) $, say $ subpath(s_1,s) $, such that it starts with a configuration in which the value of the counter is zero and there is no further such a configuration in this subpath. Let $ c_j $ be the configuration in which the counter reaches the value $ j >0 $ for the first time in $ subpath(s_1,s) $. Each such $ c $ must have a different $ t \in T $ value. Otherwise, our assumption would be violated. So, there can be at most $ |T|=M $ such $ c $'s. That is, the counter value never exceeds $ M $ in $ path(s_1,s) $, and so the length of the path can be at most $ M^{2} $ by also assuming that there is no two identical configurations in the path.
\end{proof}

%\begin{theorem}
%	Any language in $ \PCA $ can be recognized by a linear-space ptm.
%\end{theorem}
%\begin{proof}
%	Let $ \mathcal{N} = (Q,\Sigma,\delta,s_1,s_a,s_r) $ be a \pca. First of all, we show that if the value of the counter exceeds $ |Q||\tilde{x}| $ in a path, then 
%\end{proof}

%\begin{openproblem}
%	Is there any better upper bound for $ \PCA $?
%\end{openproblem}

%We continue with finding a lower bound to $ \PCA $. 

\begin{theorem}
	\label{thm:2pca-simulates-2nca}
	Let $ \mathtt{L} $ be a language recognized by a \mbox{\nca} $ \mathcal{N} $, then there exists a \mbox{\pca} $ \mathcal{P} $ recognizing $ \mathtt{L} $ with one-sided bounded-error.
\end{theorem}
\begin{proof}
	Let $ x $ be an input string.
	We begin with constructing a \pca, say $ \mathcal{P}_1 $ based on  $ \mathcal{N} $. Each nondeterministic transition of $ \mathcal{N} $ is replaced with a probabilistic one: If $ l $ is the number of nondeterministic choices, then each probabilistic choice is made with probability $ \frac{1}{l} $ in $ \mathcal{P}_1 $. Let $ k $ be the maximum number of nondeterministic choices in a single step of $ \mathcal{N} $. Due to Lemma \ref{lemma:2nca-shortest-path}, we can say that, if $ x \in \mathtt{L} $, then $ \mathcal{P}_1 $ accepts the input with a probability at least $ \left( \frac{1}{k} \right)^{c|\tilde{x}|^{2}} $ for a suitable constant $ c>1 $. As a further modification, $ \mathcal{P}_1 $ restarts the computation instead of rejecting the input. So, $ \mathcal{P}_1 $ can halt and accept the input with some probability if $ x \in \mathtt{L} $, and $ \mathcal{P}_1 $ never halts if  $ x \notin \mathtt{L} $.
	
	By using $ \mathcal{P}_1 $, we construct another \pca, say $ \mathcal{P}_2 $. At the beginning of the computation, $ \mathcal{P}_2 $ simulates $ \mathcal{P}_1 $ with probability $ \frac{3}{4} $, and executes a rejecting procedure with probability $ \frac{1}{4} $, in which $ \mathcal{P}_2 $ rejects the input exactly with probability $ \left( \frac{1}{k} \right)^{c|\tilde{x}|^{2}} $ and restarts the computation with the remaining probability. So, if $ x \notin \mathtt{L} $, then the input is never accepted, and, if $ x \in \mathtt{L} $, the accepting probability is always at least 3 times bigger than the rejecting probability. 
	
	If $ \mathcal{N} $ always halts in each path, then $ \mathcal{P}_2 $ recognizes $ \mathtt{L} $ with one-sided error bound $ \frac{1}{4} $. But, if $ \mathcal{N} $ does not halt in each branch, then $ \mathcal{P}_2 $ also enters an infinite loop in some paths. To handle this problem, we make another modification. Based on $ \mathcal{P}_2 $, we construct $ \mathcal{P} $ as follows: In each step, $ \mathcal{P} $ restarts the computation with probability $ \frac{1}{2} $, and simulates $ \mathcal{P}_{2} $ with probability $ \frac{1}{2} $. Thus, each path of $ \mathcal{P}_{2} $ can terminate with probability 1, and so $ \mathtt{L} $ is recognized by $ \mathcal{P} $ with one-sided error bound $ \frac{1}{4} $. The error bound can be reduced to any desired value by using probability amplification techniques.
\end{proof}

Remark that if a language is recognized by a \nca, then it is recognized by a \mbox{\pca} with one-sided unbounded error, vice versa. Therefore, in case of one-sided error, the language recognition power of \pca's remain the same.

\begin{corollary}
	A language is recognized by a \mbox{\nca} if and only if it is recognized by a \mbox{\pca} with one-sided bounded-error.
\end{corollary}

Moreover, due to the fact that $ \DCA \subsetneq \NCA $ \cite{Ch84} and Theorem \ref{thm:2pca-simulates-2nca}, we can say that bounded-error \pca's are more powerful than \dca's.

\begin{corollary}
	$ \DCA \subsetneq \NCA \subseteq \PCA $.
\end{corollary}

%\begin{openproblem}
%	Is two-sided bounded-error essential for \pca's in the case of language recognition?
%\end{openproblem}

Now, we turn our attention to the language recognition power of bounded-error \qcca's. We begin with the definitions of two languages: $ \twin = \{ u \# u \mid u \in \{a,b\}^{*} \} $ and $ \existentialtwin = \{ u \# v_1 \# \cdots \# v_k \mid k \geq 1, u \in \{a,b\}^{*}, v_{i} \in \{a,b\}^{*}  (1 \leq i \leq k), \mbox{and } \exists i \in \{ 1,\ldots,k \} (u = v_i ) \}  $.

{\v D}uri{\v s} and Galil \cite{DG81,DG82} showed that $ \existentialtwin $ cannot be recognized by any \dca. Moreover, Chrobak stated \cite{Ch84} that $ \existentialtwin $ does not seem to be in $ \NCA $. We show that  \mbox{\qcca}'s can recognize $ \existentialtwin $ for any error bound by using a new technique which calls a \qcfa's as a black box. In the next section, we will also show that this programming technique can also be used by \qcfa's having a pebble. 

\begin{theorem}
	\label{thm:existential-twin}
	$ \existentialtwin $ can be recognized by a \mbox{\qcca} $ \mathcal{Q} $ with bounded error.
\end{theorem}
\begin{proof}
	Recently, Yakary{\i}lmaz and Say \cite{YS10A,YS10B} showed that $ \twin $ can be recognized by any \mbox{\qcfa} for any negative one-sided error bound. Let $ \mathcal{Q}_{\twin} $ be such a \mbox{\qcfa} for error bound $ \frac{1}{5} $. (We also refer the reader to Appendix \ref{app:a-2qcfa-for-twin-language} for the details of $ \mathcal{Q}_{\twin} $.) We will use  $ \mathcal{Q}_{\twin} $ as a black box. As a special remark, $ \mathcal{Q}_{\twin} $ reads the input from left to right in an infinite loop. 
	
	Let $ x $ be an input. We assume that $ x $ is of the form $ u \# v_1 \# \cdots \# v_k $ for some $ k \geq 1 $, where $ u,v_i \in \{a,b\}^{*} $ $ (1 \leq i \leq k) $. Otherwise, it is deterministically rejected. The idea behind the algorithm is that $ \mathcal{Q} $ selects $ v_1 $, and then simulates $ \mathcal{Q}_{\twin} $ by feeding $ u\#v_1 $ as the input. $ \mathcal{Q}_{\twin} $ gives the decision of rejection only if $ u \neq v_1 $. So, whenever $ \mathcal{Q}_{\twin} $ gives the decision of rejection, then $ \mathcal{Q} $ continues by selecting $ v_2 $, and so on.  We call each such selection, in which $ \mathcal{Q} $ gives the decision of rejection,  \textit{completed}. If $ x \notin \existentialtwin $, then $ \mathcal{Q} $ can obtain $ k $ completed-selection with some nonzero probability. Otherwise, this probability becomes zero since $ \mathcal{Q} $ can obtain at most $ (k-1) $ completed-selection. So, by accepting the input with some carefully tuned probability, we can obtain the desired machine. Note that $ \mathcal{Q} $ always remembers its selection by using its counter. The pseudocode of the algorithm is given below.
\begin{equation*}
	\label{program:existentialtwin}
	\mbox{
		\begin{minipage}{0.9\textwidth}
			\footnotesize
			FOR $ i = 1 $ TO $ k $ ($ v_i $ is selected)
				\\ \hspace*{15pt}
					RUN $ \mathcal{Q}_{\twin} $ on $ x' = u \# v_{i} $
					\\ \hspace*{30pt}					
					IF $ \mathcal{Q}_{\twin} $ \underline{accepts} $ x' $ THEN \textbf{TERMINATE} FOR-LOOP
					\\ \hspace*{30pt}
					IF $ \mathcal{Q}_{\twin} $ \underline{rejects} $ x' $ AND $ i=k $ THEN \textbf{REJECT} the input
			\\
			END FOR
			\\
			\textbf{ACCEPT} $ x $ with probability $ \left( \frac{1}{5} \right)^{k} $
			\\
			\textbf{RESTART} the algorithm
		\end{minipage}			
	}
\end{equation*}
As can be seen from the pseudocode, the algorithm is actually executed in an infinite loop. We begin with analysing a single round of the algorithm. It is straightforward that if $ x \in \existentialtwin $, the input is accepted with probability $ \left( \frac{1}{5} \right)^{k} $, and it is rejected with zero probability. If $ x \notin \existentialtwin $, $ \mathcal{Q}_{\twin} $ halts with the decision of rejection with a probability at least $ \frac{4}{5} $ in each iteration of the for-loop, and so the input is rejected with a probability at least $ \left( \frac{4}{5} \right)^{k} $, and it is accepted with a probability no more than $ \left( \frac{1}{5} \right)^{k} $. Therefore, the members are accepted exactly. Since the rejecting probability is at least $ 4^{k} $ times bigger than the accepting probability in a single round, the input is rejected with a probability at least $ \frac{4}{5} $. The error bound can be easily reduced to any desired value.
\end{proof}

Note that since PTM's cannot recognize $ \twin  $ in sublogarithmic space \cite{YFSA12A}, we cannot use the same idea for \pca's. In fact, we believe that $ \existentialtwin \notin \PCA $.

Another interesting language is a unary one: $ \usquare = \{ b^{n^{2}} \mid n \geq 1 \} $. It is still not known whether $ \usquare $ can be recognized by \dca's \cite{Pet94,Pet12A}. Since \qcfa's can recognize $ \squarelanguage = \{ a^{n}b^{n^{2}} \mid n \geq 1 \} $ with negative one-sided bounded error \cite{YS10B}, we can also obtain the following result. 
\begin{theorem}
	$ \usquare $ can be recognized by a \mbox{\qcca} with bounded error.
\end{theorem}
\begin{proof}
	Let $ \mathcal{Q}_{\squarelanguage} $ be a \mbox{\qcfa} recognizing $ \squarelanguage $ with error bound $ \frac{1}{3} $. A \mbox{\qcca} can iteratively ($ i = 1,\ldots,|x| $) splits the input, say $ x $, as $ b^{i}b^{|x|-i} $ by using its counter, and then can feed $ x' = a^{i}b^{|x|} $ to $ \mathcal{Q}_{\squarelanguage} $. The pseudocode of the algorithm is given below.
	\begin{equation*}
	\label{program:usquare}
	\mbox{
		\begin{minipage}{0.9\textwidth}
			\footnotesize
			FOR $ i = 1 $ TO $ |x| $ 
				\\ \hspace*{15pt}
					RUN $ \mathcal{Q}_{\squarelanguage} $ on $ x' = a^{i}b^{|x|} $
					\\ \hspace*{30pt}					
					IF $ \mathcal{Q}_{\squarelanguage} $ \underline{accepts} $ x' $ THEN \textbf{TERMINATE} FOR-LOOP
					\\ \hspace*{30pt}
					IF $ \mathcal{Q}_{\squarelanguage} $ \underline{rejects} $ x' $ AND $ i=|x| $ THEN \textbf{REJECT} the input
			\\
			END FOR
			\\
			\textbf{ACCEPT} $ x $ with probability $ \left( \frac{1}{3} \right)^{2|x|} $
			\\
			\textbf{RESTART} the algorithm
		\end{minipage}			
	}
\end{equation*}
All the remaining details including the analysis of the algorithm is similar the algorithm given in the proof of Theorem \ref{thm:existential-twin}.
\end{proof}

%%%%%%%%%%%%%%%%%%%%%%%%%%%%%%%%%%%%%%%%%%%%%%%%%%%%%%%%%%%%%%%%%%%%%%%%%%%%%%%%%%%%%%%%%%%
%%%%%%%%%%%%%%%%%%%%%%%%%%%%%%%%%%%%%%%%%%%%%%%%%%%%%%%%%%%%%%%%%%%%%%%%%%%%%%%%%%%%%%%%%%%
%%%%%%%%%%%%%%%%%%%%%%%%%%%%%%%%%%%%%%%%%%%%%%%%%%%%%%%%%%%%%%%%%%%%%%%%%%%%%%%%%%%%%%%%%%%
\section{Pebble automata}
\label{sec:pebble-automata}
%%%%%%%%%%%%%%%%%%%%%%%%%%%%%%%%%%%%%%%%%%%%%%%%%%%%%%%%%%%%%%%%%%%%%%%%%%%%%%%%%%%%%%%%%%%
%%%%%%%%%%%%%%%%%%%%%%%%%%%%%%%%%%%%%%%%%%%%%%%%%%%%%%%%%%%%%%%%%%%%%%%%%%%%%%%%%%%%%%%%%%%
%%%%%%%%%%%%%%%%%%%%%%%%%%%%%%%%%%%%%%%%%%%%%%%%%%%%%%%%%%%%%%%%%%%%%%%%%%%%%%%%%%%%%%%%%%%

A 1-pebble finite automaton has the capability of placing a pebble to at most one tape square, of sensing whether a tape square has a pebble or not, and removing the pebble from the marked tape square. The algorithms given for \mbox{\qcca}'s in the previous section can also be implemented by 1-pebble \qcfa. In these algorithms, the counter is actually used to remember some positions on the input. A pebble can also be used in the same way by marking those positions. Therefore, we can conclude the following corollaries.

\begin{corollary}
	$ \existentialtwin $ and $ \usquare $ can be recognized by some 1-pebble \mbox{\qcfa}'s with bounded error.
\end{corollary}

By using the same idea, we can also show that $ \siamtwins = \{uu \mid u \in\{a,b\}^{*} \} $ can also be recognized by 1-pebble \qcfa's. This is an interesting language since $ \siamtwins $ cannot be recognized by any 1-pebble NTM using sublogarithmic space \cite{IIIO05}.

\begin{theorem}
	$ \existentialtwin $ can be recognized by a 1-pebble \mbox{\qcfa} $ \mathcal{Q} $ with bounded error.
\end{theorem}
\begin{proof}
	Let $ x $ be a given input. If $ x $ is not of the form $ ax_1ax_2 $ or $ bx_1bx_2 $, then it is rejected immediately, where $ x_1,x_2 \in \{a,b\}^{*} $. Assume that $ x = ax_1ax_2 $. (The other case is the same). We will use the \mbox{\qcca} algorithm given for $ \existentialtwin $ after some modifications. We give the the pseudocode of the algorithm. 
\begin{equation*}
	\label{program:siamtwins}
	\mbox{
		\begin{minipage}{0.9\textwidth}
			\footnotesize
			BEGIN (OUTER-)LOOP
				\\ \hspace*{15pt}
				TRY to MARK the next $ a $ on $ x $
					\\ \hspace*{30pt}
					IF there is no such $ a $, THEN \textbf{REJECT} $ x $
					\\ \hspace*{30pt}
					ELSE $ x = ax_1 a' x_2 $ ($ a' $ is the marked one)
				\\ \hspace*{15pt}
				BEGIN (INNER-)LOOP
					\\ \hspace*{30pt}
					RUN $ \mathcal{Q}_{\twin} $ on $ x' = x_1 \# x_2 $
					\\ \hspace*{45pt}					
					IF $ \mathcal{Q}_{\twin} $ \underline{accepts} $ x' $ THEN \textbf{TERMINATE} (OUTER-)LOOP	
					\\ \hspace*{45pt}			
					IF $ \mathcal{Q}_{\twin} $ \underline{rejects} $ x' $ 
						THEN \textbf{TERMINATE} (INNER-)LOOP
				\\ \hspace*{15pt}
				END (INNER-)LOOP
			\\
			END (OUTER-)LOOP
			\\
			\textbf{ACCEPT} $ x $ with probability $ \left( \frac{1}{5} \right)^{|x|} $
			\\
			\textbf{RESTART} the algorithm
		\end{minipage}			
	}
\end{equation*}
All the remaining details of the algorithm are similar to the previous ones.
% the algorithm given in the proof of Theorem \ref{thm:existential-twin}.
\end{proof}

Ravikumar \cite{Rav07} showed that 1-pebble \pfa's are more powerful than \mbox{\pfa}'s in the unbounded error case by giving an unbounded-error  1-pebble \mbox{\pfa} for nonstochatic language $ \centerlanguage = \{ ubv \mid u,v \in \{a,b\}^{*} \mbox{ and } |u|=|v| \} $. We show the same separation for the bounded-error machines as conjectured by Ravikumar \cite{Rav92}. (Note that bounded-error 1-pebble \mbox{\qcfa}'s can recognize $ \centerlanguage $. However, we do not know such a \mbox{\qcfa} for nonstochastic language $ \say = \{ x \mid \exists x_{1},x_{2}, y_{1}, y_{2} \in \{a,b\}^{*}, x=x_{1}bx_{2}=y_{1}by_{2}, |x_{1}| = |y_{2}| \} $ \cite{FYS10A}.)

\begin{theorem}
	1-pebble \mbox{\pfa}'s are more powerful than \mbox{\pfa}'s in the bounded error case.
\end{theorem}
\begin{proof}
	Lapin\v{s} \cite{Lap74} showed that language $ \lapins = \{ a^{m}b^{n}c^{p} \mid m^{4} > n^{2} > p > 0 \} $ is a nonstochastic, i.e. not recognized by any bounded-error \mbox{\pfa}. It is not hard to show that $ \lapins $ can be recognized by a 1-pebble bounded-error \mbox{\pfa} if there exists a bounded-error 1-pebble \mbox{\pfa} for $ \greatersquare = \{ a^{m}b^{n} \mid m > n^{2} > 0 \} $. Therefore, it is sufficient to show that there exists a bounded-error 1-pebble \mbox{\pfa}, say $ \mathcal{P} $, for $ \greatersquare $.
	
	Let $ x $ be an input string of the form $ a^{m}b^{n} $ ($ m,n>0 $). By using its pebble, $ \mathcal{P} $ can feed $ a^{m}b^{n^{2}} $, i.e. it can $ n $-times read $ b^{n} $, to any bounded-error \mbox{\pfa}. Since language $ \greater = \{ a^{m}b^{n} \mid m > n > 0 \} $ can be recognized by \mbox{\pfa}'s with bounded error \cite{Fr81,Rav92}, $ \mathcal{P} $ can recognize $ \greatersquare $ with bounded error.
\end{proof}

\textbf{Acknowledgements.}
We would like to thank A. C. Cem Say not only for his useful comments on a draft of this paper but also for many helpful discussions on the subject matter of this paper; Juraj Hromkovi\v{c} and Holger Petersen for kindly answering our questions; and, Holger Petersen for giving an idea used in the proof of Lemma \ref{lemma:2nca-shortest-path}.

\newpage
\appendix

%%%%%%%%%%%%%%%%%%%%%%%%%%%%%%%%%%%%%%%%%%%%%%%%%%%%%%%%%%%%%%%%%%%%%%%%%%%%%%%%%%%%%%%%%%%
%%%%%%%%%%%%%%%%%%%%%%%%%%%%%%%%%%%%%%%%%%%%%%%%%%%%%%%%%%%%%%%%%%%%%%%%%%%%%%%%%%%%%%%%%%%
%%%%%%%%%%%%%%%%%%%%%%%%%%%%%%%%%%%%%%%%%%%%%%%%%%%%%%%%%%%%%%%%%%%%%%%%%%%%%%%%%%%%%%%%%%%
\section{The modes of language recognition with errors}
\label{app:language-recognition-mode}
%%%%%%%%%%%%%%%%%%%%%%%%%%%%%%%%%%%%%%%%%%%%%%%%%%%%%%%%%%%%%%%%%%%%%%%%%%%%%%%%%%%%%%%%%%%
%%%%%%%%%%%%%%%%%%%%%%%%%%%%%%%%%%%%%%%%%%%%%%%%%%%%%%%%%%%%%%%%%%%%%%%%%%%%%%%%%%%%%%%%%%%
%%%%%%%%%%%%%%%%%%%%%%%%%%%%%%%%%%%%%%%%%%%%%%%%%%%%%%%%%%%%%%%%%%%%%%%%%%%%%%%%%%%%%%%%%%%

A language $ \mathtt{L} $ is said to be recognized by a machine $ \mathcal{M} $ with error bound $ \epsilon < \frac{1}{2} $, if $ \mathcal{M} $ accepts each member of $ \mathtt{L} $ with a probability at least $ 1- \epsilon $, and $ \mathcal{M} $ rejects each non-member of $ \mathtt{L} $ with a probability at least $ 1- \epsilon $. A language $ \mathtt{L} $ is said to be recognized by a machine $ \mathcal{M} $ with bounded error if it is recognized by $ \mathcal{M} $ with an error bound.
		
A language $ \mathtt{L} $ is said to be recognized by a machine $ \mathcal{M} $ with (positive) one-sided error bound $ \epsilon < 1 $, if $ \mathcal{M} $ accepts each member of $ \mathtt{L} $ with a probability at least $ 1- \epsilon $, and $ \mathcal{M} $ rejects each non-member of $ \mathtt{L} $ with probability 1. A language $ \mathtt{L} $ is said to be recognized by a machine $ \mathcal{M} $ with (positive) one-sided bounded error if it is recognized by $ \mathcal{M} $ with a positive one-sided error bound.
		
A language $ \mathtt{L} $ is said to be recognized by a machine $ \mathcal{M} $ with negative one-sided error bound $ \epsilon < 1 $, if $ \mathcal{M} $ accepts each member of $ \mathtt{L} $ with probability 1, and $ \mathcal{M} $ rejects each non-member of $ \mathtt{L} $ with a probability at least $ 1- \epsilon $. A language $ \mathtt{L} $ is said to be recognized by a machine $ \mathcal{M} $ with negative one-sided bounded error if it is recognized by $ \mathcal{M} $ with a negative one-sided error bound.
		
A language $ \mathtt{L} $ is said to be recognized by a machine $ \mathcal{M} $ with unbounded-error if $ \mathcal{M} $ accepts each member of $ \mathtt{L} $ with a probability bigger than $ \frac{1}{2} $, and $ \mathcal{M} $ rejects each non-member of $ \mathtt{L} $ with a probability at most $ \frac{1}{2} $.
		
A language $ \mathtt{L} $ is said to be recognized by a machine $ \mathcal{M} $ with one-sided unbounded-error if $ \mathcal{M} $ accepts each member of $ \mathtt{L} $ with some nonzero probability, and $ \mathcal{M} $ rejects each non-member of $ \mathtt{L} $ with probability 1.

%%%%%%%%%%%%%%%%%%%%%%%%%%%%%%%%%%%%%%%%%%%%%%%%%%%%%%%%%%%%%%%%%%%%%%%%%%%%%%%%%%%%%%%%%%%
%%%%%%%%%%%%%%%%%%%%%%%%%%%%%%%%%%%%%%%%%%%%%%%%%%%%%%%%%%%%%%%%%%%%%%%%%%%%%%%%%%%%%%%%%%%
%%%%%%%%%%%%%%%%%%%%%%%%%%%%%%%%%%%%%%%%%%%%%%%%%%%%%%%%%%%%%%%%%%%%%%%%%%%%%%%%%%%%%%%%%%%
\section{A 2qcfa for TWIN language}
\label{app:a-2qcfa-for-twin-language}
%%%%%%%%%%%%%%%%%%%%%%%%%%%%%%%%%%%%%%%%%%%%%%%%%%%%%%%%%%%%%%%%%%%%%%%%%%%%%%%%%%%%%%%%%%%
%%%%%%%%%%%%%%%%%%%%%%%%%%%%%%%%%%%%%%%%%%%%%%%%%%%%%%%%%%%%%%%%%%%%%%%%%%%%%%%%%%%%%%%%%%%
%%%%%%%%%%%%%%%%%%%%%%%%%%%%%%%%%%%%%%%%%%%%%%%%%%%%%%%%%%%%%%%%%%%%%%%%%%%%%%%%%%%%%%%%%%%

In this section, we show that there exists a \qcfa, say $ \mathcal{Q}_{\twin} $, recognizing $ \twin $ with negative one-sided  error bound $ \frac{1}{5} $.  Note that this error bound can be easily be reduced to any desired value by using probability amplification techniques.

Let $ x \in \{a,b,\# \}^{*} $ be an input. If $ x $ does not contain exactly one $ \# $, then it is deterministically rejected. So, we assume that $ x = u_1 \# u_2 $ in the following part, where $ u_1,u_2 \in \{a,b\}^{*} $.

The quantum register of $ \mathcal{Q}_{\twin} $ has three states, i.e. $ q_1$, $q_2$, and $q_3 $. In an infinite loop, $ \mathcal{Q} $ reads the input from left to right with a speed of one symbol per step. We call each iteration \textit{a round}. During scanning the input, $ \mathcal{Q}_{\twin} $ is trying to encode $ u_1 $ and $ u_2 $ into the amplitudes of $ q_1 $ and $ q_2 $, respectively. Since the encoding techniques requires to reduce the amplitudes with a constant, the current round is terminated, and then a new round is initiated with some probability after reading each symbol. If $ \mathcal{Q} $ succeeds to reach $ \dollar $, then the amplitudes of $ q_1 $ and $ q_2 $ are subtracted, based on which the input is rejected, and the input is always accepted with a small probability. So, if $ u_1 = u_2 $, the input is only accepted in each round. Otherwise, the input is both accepted and rejected. By tuning the accepting probability sufficiently small than the minimum rejecting probability, we can obtain the desired machine. The technical details are given below.

We encode the strings in base-2 and use $ 0 $ for $ a $'s and $ 1 $ for $ b $'s. Since it can be easily understandable from the context, we will use string representation also for their encodings. The strings $ 0^{j_{1}}1u $ and $ 0^{j_{2}}1u $ can be different, but, their encodings are the same, where $ j_1 $ and $ j_2 $ are nonnegative integers. Therefore, we encode $ 1u_1 $ and $ 1u_2 $ instead of $ u_1 $ and $ u_2 $. Note that $ u_1 = u_2 $ if and only if $ 1u_1 = 1u_2 $.

At the beginning of each round, the quantum register is set to 
\[
	\ket{ \psi_0 } = \left( \begin{array}{c} 1 \\ 0 \\ 0 \end{array}	 \right).
\]
Until reading $ \dollar $ symbol, the round is terminated unless the outcome 1 is observed. In order to simplify the calculations, we trace the quantum part by unconditional state vectors.

After reading $ \cent $, the following superoperator is applied:
\[
	\mathcal{E}_{\cent} = 
	\left\lbrace 
		E_1 = \frac{1}{3} 
		\left(  \begin{array}{lll}
			1 & 0 & 0 \\
			1 & 0 & 0 \\
			1 & 0 & 0
		\end{array} \right),
		~~
		E_2 = \frac{1}{3} 
		\left(  \begin{array}{lll}
			1 & 0 & 0 \\
			1 & 0 & 0 \\
			2 & 0 & 0
		\end{array} \right),
		~~
		E_3 = \frac{1}{3} 
		\left(  \begin{array}{lll}
			0 & 0 & 0 \\
			0 & 3 & 0 \\
			0 & 0 & 3
		\end{array} \right)
	\right\rbrace.
\]
Then, the (unconditional) state vector becomes
\[
	\ket{ \widetilde{ \psi_1 } } = \frac{1}{3} \left( \begin{array}{c} 1 \\ 1 \\ 1 \end{array}	 \right).
\]
Thus, the first symbol of both $ 1u_1 $ and $ 1u_2 $ have been encoded.

Until reading symbol $ \# $, the remaining part of $ 1u_1 $ is encoded into the amplitude of $ q_1 $ by using the following two superoperators. The first (second) one is applied after reading an $ a $ (a $ b $).
\[
	\mathcal{E}_{a} = 
	\left\lbrace 
		E_1 = \frac{1}{3} 
		\left(  \begin{array}{lll}
			2 & 0 & 0 \\
			0 & 1 & 0 \\
			0 & 0 & 1
		\end{array} \right),
		~~
		E_2 = \frac{1}{3} 
		\left(  \begin{array}{lll}
			2 & 0 & 0 \\
			1 & 0 & 0 \\
			0 & 0 & 0
		\end{array} \right),
		~~
		E_3 = \frac{1}{3} 
		\left(  \begin{array}{rrr}
			0 & 2 & 2 \\
			0 & 2 & -2 \\
			0 & 0 & 0
		\end{array} \right)
	\right\rbrace.
\]
\[
	\mathcal{E}_{b} = 
	\left\lbrace 
		E_1 = \frac{1}{3} 
		\left(  \begin{array}{lll}
			2 & 0 & 1 \\
			0 & 1 & 0 \\
			0 & 0 & 1
		\end{array} \right),
		~~
		E_2 = \frac{1}{3} 
		\left(  \begin{array}{rrr}
			-1 & 0 & 2 \\
			2 & 0 & 0 \\
			0 & 0 & 1
		\end{array} \right),
		~~
		E_3 = \frac{1}{3} 
		\left(  \begin{array}{rrr}
			0 & 2 & 1 \\
			0 & -2 & 1 \\
			0 & 0 & 0
		\end{array} \right)
	\right\rbrace.
\]
Then, the (unconditional) state vector becomes
\[
	\ket{ \widetilde{ \psi_1 } } = \left( \frac{1}{3} \right)^{|\cent u_1|} \left( \begin{array}{c} 1u_1 \\ 1 \\ 1 \end{array}	 \right).
\]

After reading $ \# $, the following superoperator is applied:
\[
	\mathcal{E}_{\#} = 
	\left\lbrace 
		E_1 = \frac{1}{3} 
		\left(  \begin{array}{lll}
			1 & 0 & 0 \\
			0 & 1 & 0 \\
			0 & 0 & 1
		\end{array} \right),
		~~
		E_2 = \frac{1}{3} 
		\left(  \begin{array}{lll}
			2 & 0 & 0 \\
			0 & 2 & 0 \\
			0 & 0 & 2
		\end{array} \right),
		~~
		E_3 = \frac{1}{3} 
		\left(  \begin{array}{lll}
			2 & 0 & 0 \\
			0 & 2 & 0 \\
			0 & 0 & 2
		\end{array} \right)
	\right\rbrace.
\]
Then, the (unconditional) state vector becomes
\[
	\ket{ \widetilde{ \psi_{|\cent u_1\#|} } } = \left( \frac{1}{3} \right)^{|\cent u_1 \#|} \left( \begin{array}{c} 1u_1 \\ 1 \\ 1 \end{array}	 \right).
\]

Until reading symbol $ \dollar $, the remaining part of $ 1u_2 $ is encoded into the amplitude of $ q_2 $ by using the following two superoperators. The first (second) one is applied after reading an $ a $ (a $ b $).
\[
	\mathcal{E'}_{a} = 
	\left\lbrace 
		E_1 = \frac{1}{3} 
		\left(  \begin{array}{lll}
			1 & 0 & 0 \\
			0 & 2 & 0 \\
			0 & 0 & 1
		\end{array} \right),
		~~
		E_2 = \frac{1}{3} 
		\left(  \begin{array}{rrr}
			2 & 0 & 2 \\
			2 & 0 & -2 \\
			0 & 2 & 0
		\end{array} \right),
		~~
		E_3 = \frac{1}{3} 
		\left(  \begin{array}{rrr}
			0 & 1 & 0 \\
			0 & 0 & 0 \\
			0 & 0 & 0
		\end{array} \right)
	\right\rbrace.
\]
\[
	\mathcal{E'}_{b} = 
	\left\lbrace 
		E_1 = \frac{1}{3} 
		\left(  \begin{array}{lll}
			1 & 0 & 0 \\
			0 & 2 & 1 \\
			0 & 0 & 1
		\end{array} \right),
		~~
		E_2 = \frac{1}{3} 
		\left(  \begin{array}{rrr}
			0 & 1 & -2 \\
			2 & 0 & 1 \\
			-2 & 0 & 1
		\end{array} \right),
		~~
		E_3 = \frac{1}{3} 
		\left(  \begin{array}{lll}
			0 & 0 & 1 \\
			0 & 2 & 0 \\
			0 & 0 & 0
		\end{array} \right)
	\right\rbrace.
\]
Then, the (unconditional) state vector becomes
\[
	\ket{ \widetilde{ \psi_{|\cent u_1\#u_2|} } } = \left( \frac{1}{3} \right)^{|\cent u_1\#u_2|} \left( \begin{array}{c} 1u_1 \\ 1u_2 \\ 1 \end{array}	 \right).
\]

After reading $ \dollar $, the decision on the input is given by applying the following superoperatos. If outcome 1 is observed, then the input is rejected, if output 2 is observed, then the input is accepted, and a new round is initiated, otherwise.
\[
	\mathcal{E}_{\dollar} = 
	\mspace{-4mu}
	\left\lbrace 
		\mspace{-4mu}
		E_1 = \frac{1}{3} 
		\left(  \begin{array}{rrr}
			2 & -2 & 0 \\
			0 & 0 & 0 \\
			0 & 0 & 0
		\end{array} \right)\mspace{-6mu},
		E_2 = \frac{1}{3} 
		\left(  \begin{array}{rrr}
			0 & 0 & 0 \\
			0 & 0 & 0 \\
			0 & 0 & 1
		\end{array} \right)\mspace{-6mu},
		E_3 = \frac{1}{3} 
		\left(  \begin{array}{lll}
			2 & 2 & 0 \\
			1 & 0 & 0 \\
			0 & 1 & 0
		\end{array} \right)\mspace{-6mu},
	E_4 = \frac{1}{3} 
		\left(  \begin{array}{lll}
			0 & 0 & 2 \\
			0 & 0 & 2 \\
			0 & 0 & 0
		\end{array} \right)
		\mspace{-6mu}
	\right\rbrace.
\]
If output 1 is observed, then the (unconditional) state vector becomes
\[
	\left( \frac{1}{3} \right)^{|\tilde{x}|} \left( \begin{array}{c} 2(1u_1-1u_2) \\ 0 \\ 0 \end{array}	 \right).
\]
That is, if $ u_1=u_2 $, then the input is rejected with zero probability, and if $ u_1 \neq u_2 $, the input is rejected with probability 
\[
	4 \left( \frac{1}{3} \right)^{2|\tilde{x}|}  \left( 1u_1-1u_2 \right),
\]
which can be at least 
\[
	4 \left( \frac{1}{3} \right)^{2|\tilde{x}|}.
\]
If output 2 is observed, then the (unconditional) state vector becomes
\[
	\left( \frac{1}{3} \right)^{|\tilde{x}|} \left( \begin{array}{c} 0 \\ 0 \\ 1 \end{array}	 \right).
\]
That is, the input is always rejected with probability
\[
	\left( \frac{1}{3} \right)^{2|\tilde{x}|},
\]
which is 4 times smaller than the minimum nonzero rejecting probability.

So, if $ x \in \twin $, then it is accepted with probability 1, and if $ x \notin \twin $, then it is accepted with a probability at most $ \frac{1}{5} $, and rejected with a probability at least $ \frac{4}{5} $.

\newpage

\bibliographystyle{plain}% the recommended bibstyle
\bibliography{Yakaryilmaz}

\begin{thebibliography}{10}

\bibitem{AW02}
Andris Ambainis and John Watrous.
\newblock Two--way finite automata with quantum and classical states.
\newblock {\em Theoretical Computer Science}, 287(1):299--311, 2002.

\bibitem{Ch84}
Marek Chrobak.
\newblock Nondeterminism is essential for two-way counter machines.
\newblock In {\em Proceedings of the Mathematical Foundations of Computer
  Science 1984}, pages 240--244, 1984.

\bibitem{Co93A}
Anne Condon.
\newblock {\em Complexity Theory: Current Research}, chapter The complexity of
  space bounded interactive proof systems, pages 147--190.
\newblock Cambridge University Press, 1993.

\bibitem{CL89}
Anne Condon and Richard~J. Lipton.
\newblock On the complexity of space bounded interactive proofs (extended
  abstract).
\newblock In {\em FOCS'89: Proceedings of the 30th Annual Symposium on
  Foundations of Computer Science}, pages 462--467, 1989.

\bibitem{DG81}
Pavol {\v D}uri\v{s} and Zvi Galil.
\newblock Fooling a two-way automaton or one pushdown store is better than one
  counter for two way machines (preliminary version).
\newblock In {\em STOC'81: Proceedings of the 13th Annual ACM Symposium on
  Theory of Computing}, pages 177--188, 1981.

\bibitem{DG82}
Pavol {\v D}uri\v{s} and Zvi Galil.
\newblock Fooling a two way automaton or one pushdown store is better than one
  counter for two way machines.
\newblock {\em Theoretical Computer Science}, 21:39--53, 1982.

\bibitem{DS92}
Cynthia Dwork and Larry Stockmeyer.
\newblock Finite state verifiers $\mbox{I}$: The power of interaction.
\newblock {\em Journal of the ACM}, 39(4):800--828, 1992.

\bibitem{FS89}
Uriel Feige and Adi Shamir.
\newblock Multi-oracle interactive protocols with space bounded verifiers.
\newblock In {\em Structure in Complexity Theory Conference}, pages 158--164,
  1989.

\bibitem{Fr81}
R\={u}si\c{n}\v{s} Freivalds.
\newblock Probabilistic two-way machines.
\newblock In {\em Proceedings of the International Symposium on Mathematical
  Foundations of Computer Science}, pages 33--45, 1981.

\bibitem{FYS10A}
R\={u}si\c{n}\v{s} Freivalds, Abuzer Yakary{\i}lmaz, and A.~C.~Cem Say.
\newblock A new family of nonstochastic languages.
\newblock {\em Information Processing Letters}, 110(10):410--413, 2010.

\bibitem{HS05}
Juraj Hromkovic and Georg Schnitger.
\newblock On the power of randomized multicounter machines.
\newblock {\em Theoretical Computer Science}, 330(1):135--144, 2005.

\bibitem{IIIO05}
Atsuyuki Inoue, Akira Ito, Katsushi Inoue, and Tokio Okazaki.
\newblock Some properties of one-pebble \mbox{T}uring machines with
  sublogarithmic space.
\newblock {\em Theoretical Computer Science}, 341(1-3):138--149, 2005.

\bibitem{Lap74}
J\={a}nis Lapi\c{n}\v{s}.
\newblock On nonstochastic languages obtained as the union and intersection of
  stochastic languages.
\newblock {\em Avtom. Vychisl. Tekh.}, (4):6--13, 1974.
\newblock (Russian).

\bibitem{vEB90}
Peter \mbox{van Emde Boas}.
\newblock {\em Handbook of Theoretical Computer Science (vol. \mbox{A})},
  chapter Machine models and simulations, pages 1--66.
\newblock 1990.

\bibitem{Mi61}
Marvin Minsky.
\newblock Recursive unsolvability of post's problem of ``tag'' and other topics
  in theory of \mbox{T}uring machines.
\newblock {\em Annals of Mathematics}, 74(3):437--455, 1961.

\bibitem{Mi67}
Marvin Minsky.
\newblock {\em Computation: Finite and Infinite Machines}.
\newblock Prentice-Hall, 1967.

\bibitem{Pet94}
Holger Petersen.
\newblock Two-way one-counter automata accepting bounded languages.
\newblock {\em SIGACT News}, 25(3):102--105, 1994.

\bibitem{Pet12A}
Holger Petersen.
\newblock Private communication, June 2012.

\bibitem{Rav92}
Bala Ravikumar.
\newblock Some observations on 2-way probabilistic finite automata.
\newblock In {\em FSTTCS'92: Proceedings of the 12th Conference on Foundations
  of Software Technology and Theoretical Computer Science}, pages 392--403,
  1992.

\bibitem{Rav07}
Bala Ravikumar.
\newblock On some variations of two-way probabilistic finite automata models.
\newblock {\em Theoretical Computer Science}, 376(1-2):127--136, 2007.

\bibitem{SY12A}
A.~C.~Cem Say and Abuzer Yakary{\i}lmaz.
\newblock Quantum counter automata.
\newblock {\em International Journal of Foundations of Computer Science}, (to
  appear).

\bibitem{Yak12A}
Abuzer Yakary{\i}lmaz.
\newblock Turing-equivalent automata using a fixed-size quantum memory.
\newblock Technical Report arXiv:1205.5395, 2012.

\bibitem{YFSA12A}
Abuzer Yakary{\i}lmaz, R\={u}si\c{n}\v{s} Freivalds, A.~C.~Cem Say, and Ruben
  Agadzanyan.
\newblock Quantum computation with write-only memory.
\newblock {\em Natural Computing}, 11(1):81--94, 2012.

\bibitem{YS10A}
Abuzer Yakary{\i}lmaz and A.~C.~Cem Say.
\newblock Languages recognized by nondeterministic quantum finite automata.
\newblock {\em Quantum Information and Computation}, 10(9\&10):747--770, 2010.

\bibitem{YS10B}
Abuzer Yakary{\i}lmaz and A.~C.~Cem Say.
\newblock Succinctness of two-way probabilistic and quantum finite automata.
\newblock {\em Discrete Mathematics and Theoretical Computer Science},
  12(2):19--40, 2010.

\bibitem{YS11A}
Abuzer Yakary{\i}lmaz and A.~C.~Cem Say.
\newblock Unbounded-error quantum computation with small space bounds.
\newblock {\em Information and Computation}, 279(6):873--892, 2011.

\bibitem{YKI05}
Tomohiro Yamasaki, Hirotada Kobayashi, and Hiroshi Imai.
\newblock Quantum versus deterministic counter automata.
\newblock {\em Theoretical Computer Science}, 334(1-3):275--297, 2005.

\end{thebibliography}

\end{document}